\documentclass{new_tlp}

\usepackage{mathptmx}
\usepackage{graphicx}
\usepackage{amsfonts}
\usepackage[utf8]{inputenc}
\usepackage[english]{babel}
\usepackage{amssymb}
\usepackage{etoolbox}
\usepackage{xpatch}
\usepackage{mathtools}  
\usepackage{framed}
\usepackage{stmaryrd}
\usepackage{algorithm}
\usepackage[noend]{algpseudocode}
\usepackage{multirow}
\usepackage{url}

\newcommand{\vars}{\mathit{vars}}
\newcommand{\dom}{\mathit{dom}}
\newcommand{\pred}{\mathit{def}}

\newcommand{\gen}{\ensuremath{\mathit{gen}}}

\newcommand{\compatible}{\mathit{comp}}
\newcommand{\enforce}{\triangleleft}

\newcommand{\img}{\mathit{img}}
\newcommand{\mcg}{\mathit{mcg}}

\title{Anti-unification in Constraint Logic Programming}

\author[G. Yernaux and W. Vanhoof]
{GONZAGUE YERNAUX and WIM VANHOOF\\
	University of Namur, Belgium\\
	Namur Digital Institute\\
	\email{gonzague.yernaux@unamur.be}}

\newtheorem{definition}{Definition}
\newtheorem{example}{Example}

\newtheorem{proposition}{Proposition}
\newtheorem{theorem}{Theorem}

\begin{document}         
	
	\maketitle 
	
	\begin{abstract}
		Anti-unification refers to the process of generalizing two (or more) goals into a single, more general, goal that captures some of the structure that is common to all initial goals. In general one is typically interested in computing what is often called a most specific generalization, that is a generalization that captures a maximal amount of shared structure. 
		In this work we address the problem of anti-unification in CLP, where goals can be seen as unordered sets of atoms and/or constraints. We show that while the concept of a most specific generalization can easily be defined in this context, computing it becomes an NP-complete problem. We subsequently introduce a generalization algorithm that computes a well-defined abstraction whose computation can be bound to a polynomial execution time. Initial experiments show that even a naive implementation of our algorithm produces acceptable generalizations in an efficient way. Under consideration for acceptance in TPLP.
	\end{abstract}
	
	
	\section{Introduction and motivation}
	Anti-unification refers to the process of computing for a given set of symbolic expressions $S$, a so-called \textit{generalization} of $S$, that is a single expression that captures some of the common structure that is shared by all elements in $S$. For instance, in a logic programming context, the atom $p(a,Y,f(X))$ can be seen as a generalization of the set of atoms \[\{p(a,a,f(a)), p(a,b,f(g(c))), p(a,A,f(a))\}\] as each of these atoms is an instance of $p(a,Y,f(X)$. 
	Often, one is interested in what is called a \textit{most specific} or, equivalently, a least general generalization. That is, a generalization that preserves a maximal amount of common structure. In the example above, $p(a,Y,f(X))$ is a most specific generalization of the three given atoms although other, less specific, generalizations exist such as $p(a,Y,X)$ and $p(Z,Y,X)$. 
	Being able to compute such generalizations is a mandatory ingredient in a number of program analyses and transformations such as partial deduction (e.g.~\cite{Gallagher:1993:TSL:154630.154640,DESCHREYE1999231}, supercompilation (e.g.~\cite{Sorensen:1998:IS:645795.665914}) and fold/unfold (e.g.~\cite{DBLP:journals/csur/PettorossiP98}) transformations where it is typically used as a mean to guarantee termination.

	In this work we develop a theory of generalization (or anti-unification) in the context of constraint logic programming (CLP) where - in its most declarative form - clause bodies and goals are conceptually represented by \textit{sets} of constraints and atoms. While some works exist on generalizing CLP, these typically focus on the underlying constraint domain and introduce widening operators (e.g. convex hull on $\mathbb{R}$) in order to generalize the constraint set at the \textit{semantical} level (e.g.~\cite{DBLP:journals/fuin/FioravantiPPS13}). Other existing works are targeted to a particular application such as learning constraints by generalization of samples of facts~\cite{subsumption}. 
	In contrast, we take a fundamentally different approach and focus on generalizing the \textit{syntactical representation} of the program structures to be generalized (basically conjunctions represented by sets of constraints and atoms), and this independent of the particular constraint or application domain. 
	Our main motivation for doing so is to obtain a generalization operator that computes the maximal common syntactical structure shared by two goals or, by extension, clauses and predicates. This is a basic operation needed in the work on clone detection and detection of algorithmic equivalence (see e.g.~\cite{DBLP:conf/ppdp/MesnardPV16}) where one needs to frequently and rapidly compute such generalizations in order to compare how closely related two goals or clauses are. Moreover, the generalization operator we propose
	being domain- and application-independent, it could readily be integrated in other program manipulation approaches that need to generalize CLP clauses (examples include conjunctive partial deduction or ILP-based learning). While other more involved generalization approaches exist, for example grammar-based E-generalization~\cite{DBLP:journals/corr/Burghardt14b} and regular tree abstraction~\cite{BOUAJJANI200637}, we focus in this work on the most specific generalization (msg) as it suits best our particular context.

	Computing a most specific generalization (\textit{msg}) of two or more terms (and, by extension, atoms) or other tree-like structures is straightforward and can be done in linear time. Existing algorithms are typically based on the seminal algorithm of Plotkin~\cite{plotkin} in which two tree-structures are generalized by computing their maximal common subtree and replacing non-matching subbranches by new variables.
	However, when more involved computational structures need to be generalized (such as conjunctions of atoms, goals and clauses), the literature is less clear on what algorithms are available to automatically compute their most specific generalization. The basic problem, of course, being that in this case one is not necessarily interested in viewing the structures that need to be generalized as simple tree structures as that would be too restrictive. Take for instance the conjunctions $a\wedge b\wedge c$ and $a\wedge c$; when these conjunctions are considered as trees, computing the \textit{msg} would result in $a\wedge X$ missing the fact that also $c$ is common to both conjunctions. 
	Dependent on the application at hand, usually an ad-hoc technique is introduced that most often boils down to applying the classical \textit{msg} operation to (a subset of) the atoms of both structures, usually preserving the order in which the atoms appear in the structure for efficiency reasons. This is for example the case in conjunctive partial deduction~\cite{DBLP:journals/toplas/LeuschelMS98} where conjunctions are treated as sequential structures and the abstraction operation generalizes \textit{ordered} subconjunctions. This is defensible when partially deducing Prolog programs where the order of the atoms in a conjunction \textit{is} important and usually needs to be preserved, but it nevertheless limits the possible outcomes of the generalization operation and makes it hard to transfer the approach towards other contexts where the order of the individual atoms or other computational constituents might be less important.

	While CLP is an important target in itself -- especially given its aptitude as a universal intermediate language for analysis and transformation~\cite{DBLP:journals/tplp/GangeNSSS15}, our generalization operator, basically manipulating sets of atoms, can also be beneficial in program transformation for classical (non-constraint) logic programming, as it allows to lift the restriction imposed by most of the existing generalization operators to preserve the order of the atoms in the conjunctions that are generalized.

	The paper is structured as follows. In Section~\ref{sec:preliminaries} we introduce some preliminary concepts and notation, in Section~\ref{section-algo} we introduce our main abstraction and algorithm, we evaluate our approach by means of a prototype implementation discussed in Section~\ref{section-prototype} before concluding in Section~\ref{section-conclusion}.
	
	\section{Preliminaries}\label{sec:preliminaries}
	
	\subsection{Constraint logic programming essentials}
	Let us first introduce some of the basic concepts and notations that will be used throughout the paper. A CLP program is traditionally defined~\cite{clpsurvey} over a CLP context, which is a 5-tuple $\langle X, \mathcal{V}, \mathcal{F}, \mathcal{L}, \mathcal{Q}\rangle$, where $X$ is a non-empty set of constant values, $\mathcal{V}$ is a set of variable names, $\mathcal{F}$ a set of function names, $\mathcal{L}$ is a set of constraint predicates over $X$ and $\mathcal{Q}$ a set of predicate symbols. The sets $X, \mathcal{V}, \mathcal{F}, \mathcal{L}$ and $\mathcal{Q}$ are all supposed to be disjoint sets. Symbols from $\mathcal{F}$, $\mathcal{L}$, and $\mathcal{Q}$ have an associated arity and as usual we write $f/n$ to represent a symbol $f$ having arity $n$. Given a CLP context $\mathcal{C} = \langle X, \mathcal{V}, \mathcal{F}, \mathcal{L}, \mathcal{Q}\rangle$, we can define the set of terms over $\mathcal{C}$ as $\mathcal{T}_\mathcal{C}= X \cup \mathcal{V} \cup \{f(t_1, t_2, ..., t_n) | f/n \in \mathcal{F}$ where $\forall i \in 1..n : t_i \in \mathcal{T}_\mathcal{C}\}$. 
	Likewise, the set of constraints over $\mathcal{C}$ is defined as $\mathcal{C}_\mathcal{C}=\{L(t_1, t_2, ..., t_n)\:|\:L/n \in \mathcal{L}\mbox{ and }\forall i \in 1..n : t_i \in \mathcal{T}_\mathcal{C}\}$ and the set of atoms as $\mathcal{A}_\mathcal{C}=\{p(V_1,\ldots,V_n)\:|\:p/n\in \mathcal{Q}\mbox{ and }\forall i:V_i\in\mathcal{V}\}$. 
	A goal $G\subseteq(\mathcal{C}_\mathcal{C}\cup\mathcal{A}_\mathcal{C})$ is a set of atoms and/or constraints. We will sometimes use the notion of a literal to refer to either a constraint or an atom.
	A program $P$ is then defined over a context $\mathcal{C}=\langle X, \mathcal{V}, \mathcal{F}, \mathcal{L}, \mathcal{Q}\rangle$ as a set of constraint Horn Clause definitions where each clause definition is of the form $p(V_1,\ldots,V_n)\leftarrow G$ where $p(V_1,\ldots,V_n)$ is an atom called the head of the clause with $\{V_1,\ldots, V_n\}$ all distinct variables, and $G$ a goal called the body of the clause. We will sometimes refer to a clause by $p(V_1,\ldots,V_n)\leftarrow C,B$ if we want to distinguish the set of constraints $C$ and the set of atoms $B$ in its body.
	A fact is a clause with only constraints in its body. For a predicate symbol $p$, we use $\pred(p)$ to denote the definition of $p$ in the program at hand, i.e. the set of clauses having a head atom using $p$ as predicate symbol. Without loss of generality, we suppose that all clauses defining a predicate have the same head (i.e. use the same variables to represent the arguments).
	
	In what follows we will often consider the context to be implicit and talk simply about a program and the predicates and clauses defined therein. Without loss of generality we assume that the set of constraint predicates $\mathcal{L}$ contains at least an equality relation represented by $=$. Note that in our definition of a clause, atoms contain only variables as arguments. This is by no means a limitation, as arguments can be instantiated by means of equality constraints in the clause body. 
	
	Different semantics have been defined for CLP. In our approach, we consider the declarative semantics as in~\cite{clpsurvey}. 	A constraint domain $\mathcal D$ is comprised of a set of values and an interpretation for the relational symbols used in the underlying context. 
	Given a constraint domain $\mathcal D$, a valuation is a mapping from variables to values and we say that a set of constraints $C$ is satisfiable, noted $\mathcal D\vDash C$ if there exists a valuation $v$ with $\dom(v)=\vars(C)$ such that $v(C)$ evaluates to $true$. In this work we focus on the declarative semantics of a program which is defined as a subset of $\mathcal{B}_\mathcal{D}$, the latter defined as $\{p(v_1,\ldots,v_n)\:|\: p/n\in \mathcal{Q}\mbox{ and } v_i\in\mathcal{D}\}$.
	For a program $P$ and an underlying constraint domain $\mathcal{D}$, the immediate consequence operator $T_P^\mathcal D$ can be defined as a continouous function on $\mathcal{B}_\mathcal{D}$ as follows~\cite{clp-semantics}:
	\[T_P^\mathcal D(I)=\left\{\begin{array}{ll|ll}
	& && p(V_1,\ldots,V_n)\leftarrow C,B\mbox{ a renamed apart clause in $P$}\\
	p(v_1,\ldots,v_n)  & & &  v\mbox{ a valuation on $\mathcal{D}$ such that $\mathcal{D}\vDash v(C)$ and $v(B)\subseteq I$}\\
	& & &\forall i\in\{1,\ldots,n\} : v_1 = v(V_i)\\
	\end{array}\right\}
	\]
	
	The semantics of a program $P$, which we will represent by $\llbracket P\rrbracket$ can then be defined as the least fixed point of $T_P^{\mathcal{D}}$. In what follows, we will often simply refer to the semantics of a program without specifying the underlying constraint domain or CLP context.
	The semantics of a goal $G$ with respect to a program $P$ and a set of variables $V=\{V_1,\ldots,V_k\}$ occurring in $G$ is then defined as $\{q_P(v_1,\ldots,v_k)\in \llbracket P'\rrbracket\}$ where $P'$ is the program $P$ to which a clause $q_P(V_1,\ldots,V_k)\leftarrow G$ has been added with $q_P$ a special predicate symbol not occurring in $P$. Slightly abusing notation, we will use $\llbracket G\rrbracket^P_{V}$ to denote the semantics of the goal $G$ w.r.t. the program $P$ and the set of variables $V$, or simply $\llbracket G\rrbracket_V$ if the program is clear from the context.
	While in practice CLP is typically used over a concrete domain, we will make abstraction of the concrete domain over which the constraints are expressed, as our generalization theory only considers the syntactical structure of the constraints (and not their semantics).

	\subsection{Generalization principles}\label{subsection-gen}
	For any program expression $e$ (be it a term, a constraint, an atom or a goal), we use $\vars(e)$ to denote the set of variables that appear in $e$. As usual, a substitution is a mapping from variables to terms and will be denoted by a Greek letter. For any mapping $\sigma$, $\dom(\sigma)$ represents its domain, $\img(\sigma)$ its image, and for a program expression $e$ and a substitution $\sigma$, $e\sigma$ represents the result of simultaneously replacing in $e$ those variables $V$ that are in $\dom(\sigma)$ by $\sigma(V)$. A renaming is a special kind of substitution, mapping variables to distinct variables (i.e. being injective). For a renaming $\rho$, we use $\rho^{-1}$ to denote its reverse. Two expressions $e_1$ and $e_2$ are variants if and only if $e_1\rho=e_2$ and $e_1=e_2\rho^{-1}$ for some renaming $\rho$. For an expression $e$, a fresh renaming of $e$ is a variant of $e$ where all variables have been renamed to new, previously unused variables. 
	Given the notion of a renaming, we can easily define a quasi-order relation between goals as follows.
	
	\begin{definition}[Generalization]\label{def_gen}
		Let $G$ and $G'$ be goals. We say that $G$ is \emph{more general than} (or, synonymously, is a generalization of) $G'$, denoted $G \preceq G'$, if and only if there exists a renaming $\rho$ such that $G\rho \subseteq G'$.
	\end{definition}
	
	Hence, a goal is more general than another goal if the former is a subset of the latter modulo a variable renaming. While our notion of generalization is simple and purely of syntactic nature, it is in line with what one could consider to be a generalization at the semantic level, since generalizing a goal corresponds to removing computational units (constraints or atoms). 
	
	\begin{example}
		Consider the goal $G = \{p(X,Y), X=a, Y=b\}$. Then the goals $\{p(X,Y), X=a\}$, $\{p(X,Y)\}$, $\{p(A,B)\}$ and $\{p(X,Y),Y=b\}$ are all generalizations of $G$.
	\end{example}
	
	In a more traditional logic programming context, an atom is typically defined as more general than another atom if the latter can be obtained from the former by applying a substitution~\cite{benkerimi:lloyd,Sorensen95analgorithm} and generalizing an atom is done by replacing terms with new variables. Since in our context, atoms are represented in simple form (i.e. all arguments being variables), the same effect can be obtained by removing constraints from the goal. Note that our definition is, at the same time, more general, as it allows to generalize a goal also by removing atoms. In a traditional logic programming context where goals are conjunctions of atoms, one need to resolve to higher-order generalization techniques in order to obtain the same effect. Also observe that in our generalization scheme, constants and functors are impossible to generalize through variabilization, because renamings are mappings from variables to variables only. 
	This is a fundamental difference of relation $\preceq$ with the $\theta$-subsumption relation of~\cite{plotkin}, the latter being defined by substitutions rather than renamings. Our relation is a first-order generalization (higher-order terms as well as predicate names can't be generalized) with firm constants and functors.
	
	Defining generalizations with injective mappings (i.e. renamings) rather than arbitrary mappings from $\mathcal{V}$ to $\mathcal{V}$ as in $\theta$-subsumption ensures that some variable $V$ cannot be generalized by two (or more) distinct variables in the computed generalization. If renamings weren't injective, a generalization could have many more variables than the goals it generalizes; in that case, the generalization could contain variables that are no longer linked on the semantic level such as new variables occurring only once. For many domains, the injective property makes more sense, not allowing variables to lose their semantics once generalized.
	
	\begin{example}
		Let us consider $G = \{X > 2, X < 10\}$ where we suppose the constraints are over some numerical domain. In our framework, the three following generalizations are correct: $\{A > 2, A < 10\}$, $\{A > 2\}$, $\{A < 10\}$. Without the restriction to injective renamings, $\{A > 2, B < 10\}$ would also be a valid generalization.
	\end{example}
	
	In practice, some domain-specific constraint predicates and functional operators could be characterized as commutative (such as $=$ and $+$ for numeric instances), which would affect their generalizations. The approach presented in this paper could easily be extended to take this property into account, but for the sake of clarity we will keep the approach purely syntactic on that point of view, only considering non-commutative symbols in textual representations of constraints. 
	Despite their differences, our generalization relation shares the following property with the usual $\theta$-subsumption order from~\cite{subsumption-lattice}.
	
	\begin{proposition}\label{prop-partial-order}
		The generalization relation $\preceq$ is a quasi-order.
	\end{proposition}
	\begin{proof}
		We need to prove that $\preceq$ is transitive and reflexive. Reflexity is immediate since for any goal $G\subseteq G$ and, thus, $G\preceq G$. 
		For transitivity, consider three arbitrary goals $G_1$, $G_2$ and $G_3$ such that $G_1 \preceq G_2$ and $G_2 \preceq G_3$. Then by definition~\ref{def_gen}, there exist $\Delta_1, \Delta_2, \rho_1$ and $\rho_2$ such that $G_2 = G_1\rho_1 \cup \Delta_1$ and $G_3 = G_2\rho_2 \cup \Delta_2$. Or, equivalently, 
		\begin{gather*}
		G_3 = (G_1\rho_1 \cup \Delta_1)\rho_2 \cup \Delta_2 = 	G_1\rho_1\rho_2 \cup (\Delta_1\rho_2 \cup \Delta_2) 
		\end{gather*}
		Since the composition of two renamings is a renaming, and the union of two sets a set, it follows that $G_1\preceq G_3$.
	\end{proof}
	
	Generalized goals are linked by their semantics as stated in Proposition~\ref{prop-semantics} below.
	
	\begin{proposition}\label{prop-semantics}
		Let $P$ be a program and $G$ and $G'$ goals. If $G\preceq G'$ such that $G\rho\subseteq G'$ for some renaming $\rho$, then for any set of variables $V\subseteq\vars(G\rho)$, we have that $\llbracket G\rho\rrbracket_V\supseteq\llbracket G'\rrbracket_{V}$.
	\end{proposition}
	\begin{proof}
		The proof is trivial given that $G\rho\subseteq G'$. Indeed, suppose that $G'$ is composed of a set of constraints $C'$ and a set of atoms $B'$. Then, if $v$ is a valuation on the underlying domain $\mathcal{D}$ such that $\mathcal{D}\vDash v(C')$ and $v(B')\subseteq \llbracket P\rrbracket$, then there exist some predicate symbol $q$ such that $q(v(V_1),\ldots,v(V_k))\in \llbracket G'\rrbracket$. Now, since $G'=G\rho\cup\Delta$ for some set of constraints and/or atoms $\Delta$, it holds that $\mathcal{D}\vDash v(C)$ for the constraints $C\subseteq G\rho$ and $v(B)\subseteq \llbracket P\rrbracket$ for the set of atoms $B\subseteq G\rho$. Consequently, $q(v(V_1),\ldots,v(V_k))\in \llbracket G\rho\rrbracket$.
	\end{proof}
	
	We can now define the computational structure that is shared by a set of goals through the concept of common generalization.  
	
	\begin{definition}[Common generalization] \label{def-common-generalization}
		Let $\{G_1, G_2, ..., G_k\}$ be a set of goals. Then a goal $G$ is a \emph{common generalization} of $\{G_1, G_2, ..., G_k\}$ if and only if $\forall i \in 1..k : G \preceq G_i$. 
	\end{definition}

	In what follows we will mostly consider common generalizations of \textit{two} goals. Note that at least one common generalization exists for any two goals: the empty set which can be seen as the most general generalization, i.e. the minimal element in the quasi-order $\preceq$. But obviously the empty set is not an interesting generalization to express similarities in groups of literals. In what follows, we are interested in computing what we call a \textit{most specific generalization}, that is a maximal element with respect to $\preceq$. A most specific generalization is also sometimes called a \textit{least general generalization}. 
	
	\begin{definition}[msg]\label{def-msg}
		Let $G$ be a common generalization of $S = \{G_1, G_2, ..., G_n\}$. Then $G$ is a \emph{most specific generalization} (msg) of $S$ if there does not exist another common generalization of $S$, say $G'$, such that $G \preceq G'$ and $G' \npreceq G$.
	\end{definition}
	
	Note that, by definition, a common generalization of two goals $G_1$ and $G_2$ is a variant of both a subset from $G_1$ and of a subset from $G_2$. Without loss of generality, we will often consider a common generalization to be a subset of one of the goals, as in the following example.
	
	\begin{example}\label{ex-common-gen}
		Let us consider the goals
		\begin{gather*}
		G_1 = \{f(X),g(X),g(Y)\} \qquad G_2 = \{f(R), g(T)\} 
		\end{gather*}
		
		$G = \{f(X), g(Y)\} \subseteq G_1$ is a common generalization of $\{G_1, G_2\}$, as there exists $\rho = [X\leftarrow R, Y \leftarrow T]$ such that $G_2 = G\rho$, so $G\preceq G_2$; it also holds that $G \subset G_1$, so $G \preceq G_1$. Moreover, $G$ is an msg of $\{G_1, G_2\}$ as no strictly less general common generalization exists, $G$ having generalized all literals in $G_2$. Note that $G_2$ is also an msg of $\{G_1, G_2\}$, which can as easily be proved. In fact, by Definition~\ref{def_gen}, any variant of $G$ is also an msg for $G_1$ and $G_2$.
	\end{example}
	
	Contrary to the case of traditional logic programming, where the most specific generalization of two goals is unique (modulo a variable renaming)~\cite{benkerimi:lloyd}, in our context two goals may typically have several most specific generalizations. 
	
	\begin{example}
		Let us consider the goals 
		\begin{gather*}
		G_1 = \{f(X), g(Y), h(X,Y)\} \qquad G_2 = \{f(R), g(U), h(T, S)\}
		\end{gather*}
		$\{f(X), g(Y)\}$ and $\{h(X,Y)\}$ are both msgs of $\{G_1, G_2\}$. Indeed, each of these generalizations doesn't allow the addition of any more literals while remaining a valid common generalization of $G_1$ and $G_2$, due to the injectivity of the generalization renamings. The two msgs are thus incomparable, $\preceq$-wise. 
	\end{example}
	
	Amongst the msgs of a set of goals, some generalizations could only have a few literals, thereby capturing less common structure than others. Ideally, we are interested in those most specific generalizations that are of maximal cardinality. 
	
	\begin{definition}[mcg]\label{def-mcg}
		Let $G$ be a common generalization of $S = \{G_1, G_2, ..., G_n\}$. Then $G$ is a \emph{maximal common generalization} (mcg) of $S$ if there does not exist another common generalization of $S$, say $G'$, such that $|G'| > |G|$.
	\end{definition}
	
	It is trivial to show that a maximal generalization $G$ of a set of goals $S$ is also a most specific generalization of $S$. Indeed, if it weren't the case, it would, by Definition~\ref{def-msg}, be possible to add some literal to $G$ and get a more specific generalization. But the latter generalization would have strictly greater cardinality than $G$, so $G$ cannot be maximal. 
	However, computing a maximal common generalization is an intractable problem.
	The reason is, of course, due to the fact that we need to match unordered \textit{sets} of literals rather than sequences, whereas the classical subsumption-based formulation from~\cite{plotkin} is computable in polynomial time. 

	In order to show this formally, we define a decision problem variant which we name MCGP (Maximal Common Generalization Problem) that we show to be NP-complete. The decision problem variant MCGP boils down to verifying whether there exist a renaming $\rho$ such that the smallest of two goals is in itself a maximal common generalization of both. Formally: given two goals $G_1$ and $G_2$ with $|G_1| \leq |G_2|$ and $\vars(G_1)\cap\vars(G_2)=\emptyset$, verify whether there exists $\rho$ such that $G_1\rho$ is a subset of $G_2$.
	
	\begin{theorem}\label{theo:NPcomplete}
		The MCGP problem is NP-complete. 
	\end{theorem}
	\begin{proof*}
		It is easy to see that MCGP is in NP: given renamed apart goals $G_1$ and $G_2$ as well as a renaming $\rho$, the application of $\rho$ on all the literals in $G_1$ will either yield a subset of $G_2$ or not, which can be verified in polynomial time.
		
		We will now perform a reduction from the Induced Subgraph Isomorphism Problem (ISIP) which is stated as follows~\cite{SYSLO198291}. Given two unoriented and unweighted graphs, $(V_1, E_1)$ and $(V_2, E_2)$, where for each graph $(V_i,E_i)$, $V_i$ denotes the set of vertices and $E_i$ the set of edges between vertices from $V_i$. Assuming, moreover, that $|V_1| \le |V_2|$, then ISIP is the problem of deciding whether $(V_1,E_1)$ is isomorphic to an induced subgraph of $(V_2,E_2)$ meaning there exists a (total) injective function $f:V_1\mapsto V_2$ such that $\forall x,y\in V_1$, there is an edge $(x,y)\in E_1$ if and only if there is an edge $(f(x),f(y))\in E_2$. 
		The problem is known to be NP-complete~\cite{SYSLO198291}.
		
		We can transform any instance of ISIP into an instance of MCGP as follows. Given the graphs $(V_1, E_1)$ and $(V_2, E_2)$ (with $|V_1|\le|V_2|$), we define goals 
		\begin{gather*}
		G_1 = \{ node(V_x)\: |\: x \in V_1 \} \cup \{ edge(V_x,V_y)\: |\: (x,y) \in E_1 \} \\
		G_2 = \{ node(V_x)\: |\: x \in V_2 \} \cup \{ edge(V_x,V_y)\: |\: (x,y) \in E_2 \}
		\end{gather*}
		
		In these goals, we suppose that \textit{node} is a unary predicate representing nodes and \textit{edge} a binary predicate representing edges between nodes. Given a node $x$ we use a variable named $V_x$ to represent this node in the goal. If $G_1$ and $G_2$ have at least one variable's name in common, considering a renamed apart version of $G_1$ rather than $G_1$ itself will ensure that the obtained instance of MCGP is valid. Using this scheme, the transformation from graphs into goals can obviously be done in polynomial time. We will now prove that this transformation preserves the positive and negative instances of ISIP, that is $(V_1,E_1)$ is isomorphic to an induced subgraph of $(V_2,E_2)$ if and only if $G_1$ is an mcg of $\{G_1,G_2\}$. 
		
		\begin{itemize}
			\item[$(\Rightarrow)$] 	Let us suppose that $(V_1,E_1)$ is 	isomorphic to an induced subgraph of $(V_2,E_2)$. In other words there exists an injective function $f:V_1\mapsto V_2$ such that $\forall x,y\in V_1$, there is an edge $(x,y)\in E_1$ if and only if there is an edge $(f(x),f(y))\in E_2$. We have to show that $G_1$ is an mcg of $G_1$ and $G_2$. Obviously the existence of $f$ implies the existence of a renaming $\rho:\vars(G_1)\mapsto\vars(G_2)$ defined as $\rho=\{(V_x,V_y)\:|\:(x,y)\in f\}$. Since $f$ is a total injective function, we have that for each $\mathit{node}(V_x)\in G_1$ there is $\mathit{node}(V_x\rho)\in G_2$ and, by definition of $f$, for each $\mathit{edge}(V_x,V_y)\in E_1$ there is $\mathit{edge}(V_x\rho,V_y\rho)\in G_2$. In other words $G_1\rho$ is a subset of $G_2$ and, hence, $G_1$ is a generalization of $G_2$ and, consequently, a maximal common generalization of $\{G_1,G_2\}$.
			
			\item[$(\Leftarrow)$] The other way round, suppose that $G_1$ 	is an mcg for $\{G_1, G_2\}$, implying there exists a renaming $\rho$ such that $G_1\rho\subseteq G_2$. Given that $\dom(\rho)=\vars(G_1)$ and that $\rho$ is injective by definition, we can define a function $f:V_1\mapsto V_2$ as $f = \{(x,y)\:|\:(V_x,V_y)\in\rho\}$ that is injective as well. Now, $\dom(f)=V_1$ (i.e. $f$ is total) since there is a $\mathit{node}(V_x)\in G_1$ for each vertex $x\in V_1$. Moreover, since $G_1\rho\subseteq G_2$, we have that for each $\mathit{edge}(V_x,V_y)\in G_1$ there exists $\mathit{edge}(V_x\rho,V_y\rho)$ and, consequently, we have that $\forall x,y\in V_1$, there is an edge $(x,y)\in E_1$ if and only if there is an edge $(f(x),f(y))\in E_2$ concluding the proof that $G_1$ is isomorphic to an induced subset of $G_2$.\mathproofbox
		\end{itemize}
	\end{proof*}
	
	\section{Anti-unification algorithm}\label{section-algo}
	
	In the following we restrict ourselves to generalizations of two renamed apart goals - each of them being a set of literals. To construct a generalization of goals $G_1$ and $G_2$ our algorithm basically needs to search for a subset of $G_1$ that is also a subset of $G_2$ (modulo a variable renaming) and vice versa. To represent these matching subsets, the algorithm will use an injective mapping $\phi\subseteq G_1\times G_2$ that associates literals from $G_1$ to matching literals of $G_2$. For such $\phi$ to represent a generalization, there must exist a renaming $\rho$ such that $\dom(\phi)\rho = \img(\phi)$ and, likewise, $\img(\phi)\rho^{-1}=\dom(\phi)$. In what follows we will use the word generalization to refer to such a mapping $\phi$ as well as to the goal(s) it represents.
	
	\begin{example}\label{ex-mapping}
		Let us consider the goals
		\begin{gather*}
		G_1 = \{f(X), f(Z), g(X, Y), h(Y, Z)\} \qquad G_2 = \{f(R), g(R,T), h(T,U), f(U)\}.
		\end{gather*} 
		Then the mapping $\phi = \{(f(X), f(R)), (g(X,Y), g(R,T))\}$ (mapping $f(X)$ from $G_1$ to $f(R)$ from $G_2$ and $g(X,Y)$ from $G_1$ to $g(R,T)$ from $G_2$) is a generalization of $G_1$ and $G_2$. Indeed, $\dom(\phi)=\{f(X), g(X,Y)\}\subseteq G_1$ and is a variant of $\img(\phi)=\{f(R),g(R,T)\}\subseteq G_2$. 
	\end{example}
	
	Since computing maximal common generalizations is an NP-complete problem, we will rather focus on computing common generalizations $\phi$ that are not necessarily maximal, but whose size is \textit{stable} in the sense that replacing a limited number of elements in $\phi$ does not give rise to a larger generalization. Let us first define the notion of a \textit{$k$-swap}, being a replacement of at most $k$ elements in a generalization.
	
	\begin{definition}[k-swap]
		Let $G_1$ and $G_2$ be two renamed apart goals, and $\phi,\phi'\subseteq G_1\times G_2$ generalizations. We say that $\phi'$ is a \emph{$k$-swap} of $\phi$ if and only if $|\phi|=|\phi'|$ and $|\phi\cap\phi'|\ge |\phi|-k$ for some $k\in\mathtt{N}$.
	\end{definition}
	
	Intuitively, a k-swap of a generalization $\phi$ is obtained from $\phi$ by changing at most $k$ elements such that the result is still a generalization.
	
	\begin{example}\label{ex-k-swap}
		Let us reconsider the generalization $\phi$ from Example~\ref{ex-mapping}. Then the generalization
		\[\phi' = \{(g(Y,X),g(R,T)), (h(Y,Z),h(T,U))\}\] is a 1-swap of $\phi$, since effectively one element has been replaced in $\phi$ to get $\phi'$. In a similar way, $\phi'' = \{(f(Z), f(U)), (h(Y,Z), h(T,U))\}$ is a 2-swap of $\phi$ (but is not a 1-swap, as two elements have been replaced to get $\phi''$). 
	\end{example}
	
	Central to our approach to get a workable anti-unification algorithm is the notion of $k$-swap \textit{stability}. We call a generalization $\phi$ of goals $G_1$ and $G_2$ $k$-swap stable if any larger generalization of these goals differs from $\phi$ in \textit{at least} $k+1$ elements.
	
	\begin{definition}[k-swap stability]\label{def-k-swap-stable}
		Let $G_1$ and $G_2$ be two renamed apart goals and $\phi\subseteq G_1\times G_2$ a generalization of $G_1$ and $G_2$. Then the generalization $\phi$ is \emph{$k$-swap stable} if and only if there does not exist a larger generalization $\hat{\phi}\supset\phi'$ where $\phi'$ is a $k$-swap of $\phi$. Such a $\hat{\phi}$ is called a \emph{k-swap extension} of $\phi$. 
	\end{definition}
	
	A $k$-swap stable generalization, even though not necessarily maximal, is at least \textit{stable} in the sense that there is no obvious way (i.e. by replacing $k$ or less elements) in which a larger generalization could be obtained. Put differently, when a generalization is constructed by a search algorithm, $k$-swap stability implies that in order to find a larger generalization, the algorithm would need to reconsider at least $k+1$ choices that were made during construction.
	
	\begin{example}\label{ex-k-swap-stable}
		Consider $G_1 = \{a(X,Y,Z), b(X), c(Z), d(Z)\}$ and $G_2 = \{a(A,B,C), a(C,B,A), b(C), c(A), d(C)\}$. Then, when $\phi$ is constructed by mapping $a(X,Y,Z)$ to $a(A,B,C)$, the largest generalization mapping that $\phi$ can grow to is $\{(a(X,Y,Z),a(A,B,C)), (d(Z),d(C))\}$ or, equivalently, the generalization $\{a(X,Y,Z),d(Z)\}$. However $\phi$ is not 1-swap stable. Indeed, mapping $a(X,Y,Z)$ to $a(C,B,A)$ instead would give rise to $\{(a(X,Y,Z),a(C,B,A)), (b(X),b(C)), (c(Z),c(A))\}$ or, equivalently, the larger generalization $\{a(X,Y,Z),b(X),c(Z)\}$.
	\end{example}
	
	Obviously, if a generalization $\phi$ between goals $G_1$ and $G_2$ is k-swap stable for all $k\in\mathbb{N}$, then $\phi$ is a maximal and thus most-specific generalization. This is in line with the intuition that as $k$ grows, any k-swap-stable generalization has increased stability and thus increased accuracy (in number of generalized literals).
	
	One more concept needs to be introduced before we can define our algorithm for computing $k$-swap stable generalizations, namely an operator that allows to combine two generalizations into a single generalization.
	
	\begin{definition}[Enforcement operator]
		Let $G_1$ and $G_2$ be two renamed apart goals. The \emph{enforcement operator} is defined as the function $\enforce: (G_1\times G_2)^2 \mapsto (G_1\times G_2)$ such that for two generalizations $\phi$ and $\phi'$ for $G_1$ and $G_2$, $\phi \enforce \phi' = \phi' \cup M$ where $M$ is the largest subset of $\phi$ such that $\phi'\cup M$ is a generalization of $G_1$ and $G_2$. 
	\end{definition}
	
	In other words, $\phi\enforce\phi'$ is the mapping obtained from $\phi\cup\phi'$ by eliminating those pairs of literals $(A,A')$ from $\phi$ that are \textit{incompatible} with some $(B,B')\in\phi'$ either because it concerns the same literal(s) or because the involved renamings cannot be combined into a single renaming. 
	
	\begin{example}
		Consider $\phi = \{(a(X, Y), a(A, B)), (b(X), b(A))\}$, a generalization of two goals $G_1$ and $G_2$. Suppose $\phi' = \{(c(Y), c(C))\}$ is also a generalization of $G_1$ and $G_2$. Enforcing $\phi'$ gives $\phi\enforce\phi' = \{(b(X), b(A)), (c(Y), c(C))\}$. Indeed, this can be seen as forcing $Y$ to be mapped on $C$; therefore the resulting generalization can no longer contain $(a(X, Y), a(A, B))$ as the latter maps $Y$ on $B$. 
	\end{example}
	
	Algorithm~\ref{alg:kswap} represents the high-level construction of a k-swap stable generalization of goals $G_1$ and $G_2$. In the algorithm, we use $\gen(G_1, G_2)$ to represent those literals from $G_1$ and $G_2$ that are variants of each other, formally $\gen(G_1,G_2)=\{(A,A')\:|\:A\in G_1, A'\in G_2\mbox{ and }A\rho=A'\mbox{ for some renaming }\rho\}$. In each round, the algorithm tries to transform the current generalization $\phi$ (which initially is empty) into a larger generalization by forcing a new pair of literals $(A,A')$ from $\gen(G_1,G_2)$ in $\phi$, which is only accepted if doing so requires to swap no more than $k$ elements in $\phi$. More precisely, the algorithm selects a subset of $\phi$ (namely $\phi_s$) that can be swapped with a subset $\phi_G$ of the remaining mappings from $\gen(G_1,G_2)$ that are not yet used such that the result of replacing $\phi_s$ by $\phi_G$ in $\phi$ and adding $(A,A')$ constitutes a generalization. Note how condition 1 in the algorithm expresses that $\phi_s$ must include at least those elements from $\phi$ that are not compatible with $(A,A')$.
	The search continues until no such $(A,A')$ can be added.
	
	\begin{algorithm}[hbtp]
		\caption{Computing a $k$-swap stable generalization $\phi$ for goals $G_1$ and $G_2$}
		\label{alg:kswap}
		\begin{algorithmic}
			\State $\phi\gets\emptyset$
			\Repeat
			\State select $(A, A')\in \gen(G_1, G_2)\setminus\phi, \phi_s 	\subseteq \phi, \phi_G\subseteq (G_1\times G_2)\setminus(\phi \cup\{(A,A')\})$ such that:
			\State \hspace*{3ex}(1) $\phi_s\supseteq\phi\setminus 	\phi\enforce\{(A,A')\}$
			\State \hspace*{3ex}(2) $|\phi_s| \le k$
			\State \hspace*{3ex}(3) $|\phi_G| = |\phi_s|$
			\State \hspace*{3ex}(4) $\phi\setminus\phi_s\cup 	\phi_G\cup\{A,A'\}$ is a generalization of $G_1$ and $G_2$
			\If{such $(A,A'), \phi_G, \phi_s$ are found}
			\State $\phi\gets \phi\setminus\phi_s\cup \phi_G \cup 	\{(A,A')\}$
			\EndIf
			\Until no such $(A,A'), \phi_G, \phi_s$ are found
		\end{algorithmic}
	\end{algorithm}	
	
	Even if the algorithm as formulated is non-deterministic and does not specify how $(A,A')$, $\phi_s$ or $\phi_G$ are computed (we will come back to this), it can easily be seen that it computes a generalization that is $k$-swap stable.
	
	\begin{theorem}\label{theo:kswapstable}
		Given renamed apart goals $G_1$, $G_2$ and a constant $k\in\mathbb{N}$, 
		the generalization computed by Algorithm~\ref{alg:kswap} is $k$-swap stable.
	\end{theorem}
	\begin{proof}
		Given goals $G_1$, $G_2$ and constant $k\in\mathbb{N}$, Algorithm~\ref{alg:kswap} can be seen as computing a sequence of generalizations $\phi^0,\ldots,\phi^n$ where each $(\phi^i)$ represents the value of $\phi$ at the end of the $i$-th loop iteration. The generalization $\phi$ is then the final value in this sequence, i.e. $\phi=\phi^n$. 
		
		The proof is by contradiction. Suppose that $\phi=\phi^n$ is not k-swap stable. By definition, this means that there exists a k-swap extension $\phi_k$ of $\phi$ such that $|\phi_k|>|\phi|$ and $\phi_k\supset\phi'$, with $\phi'$ a k-swap of $\phi$. 
		Consequently, there exist generalizations $\phi_s$, $\phi_s'$ and $\phi_r$ such that $\phi'=(\phi\setminus\phi_s)\cup\phi_s'$ and $\phi'=\phi_k\setminus\phi_r$, with $|\phi_s|=|\phi_s'|\le k$ and $|\phi_r|\ge 1$. 
		Then, by taking $\phi_G=\phi_s'$ and $(A,A')\in\phi_r$ the conditions of in the algorithm are satisfied, contradicting the fact that the algorithm's execution would end with $\phi^n$. 
	\end{proof}
	
	For a given value of $k$, Algorithm~\ref{alg:kswap} computes thus a $k$-swap stable generalization, at least if an exhaustive search is performed in each round of the repeat loop in order to find a couple $(\phi_s,\phi_G)$ that allows to transform $\phi$ into a strictly larger generalization $(\phi\setminus\phi_s)\cup\phi_G\cup\{(A,A')\}$. Even if this exhaustive search is implemented, it is not hard to see that for a given and constant value of $k$, the algorithm executes in time $\mathcal{O}(M^{ck})$, where $c$ is a constant and $M$ proportional to $|\gen(G_1,G_2)|$. Note how the exponent depends on $k$, which is a constant parameter unrelated to the size of the goals to generalize (the input). Therefore the execution time of the algorithm is polynomially bounded. 
	
	By aiming to improve some initial solution at each iteration, Algorithm~\ref{alg:kswap} is an anytime algorithm: as such, in concrete implementations one could retrieve the $n$-th generalization computed by Algorithm~\ref{alg:kswap} when it is interrupted at iteration $(n+1)$. The $n$-th generalization may not be k-swap stable, but it is assured to be a generalization of size $n$.
	Also note that being inherently non-deterministic, the algorithm is by no means guaranteed to find the largest, or most convenient, $k$-swap stable generalization. In order to somewhat steer the search towards a \textit{promising} generalization, we introduce the concept of a \textit{quality estimator}, i.e. a function that associates a value in $\mathbb{R}$ to any couple of matching literals $(A,A')\in\gen(G_1,G_2)$. The general idea behind this function being that the higher the value associated to a couple $(A,A')$, the higher the probability that $(A,A')$ is an element of a maximal common generalization. 
	
	\begin{definition}[Quality estimator]
		Given goals $G_1$ and $G_2$, a \emph{quality estimator} is a function $\Omega^{G_1, G_2} : \gen(G_1, G_2) \mapsto \mathbb{R}$. When goals $G_1$ and $G_2$ are unambiguously identified, we will simply write $\Omega$. 
	\end{definition}
	
	A typical implementation of Algorithm~\ref{alg:kswap} will thus loop through the potential couples $(A,A') \in \gen(G_1,G_2)$ in descending order of their $\Omega$-values. If $\Omega$ is a perfect oracle -- in the sense that it associates maximal values to those couples that constitute an mcg -- then, obviously, Algorithm~\ref{alg:kswap} computes this mcg. In practice, however, $\Omega$ will be a heuristic. In our implementation, which we elaborate on in Section~\ref{section-prototype}, we use the following heuristic $\Omega$-function.
	
	\begin{example}\label{ex-omega}
		An intuitive yet efficient quality estimator is the function that maps a couple $(A,A')$ to the multiplicative inverse of the number of conflicts the couple has with other couples (i.e. the involved renamings being incompatible). Let $\mathfrak{c}$ denote the set $\big\{(B,B')\in gen(G_1, G_2) | (B,B')\neq(A,A')\wedge \{(A,A'), (B,B')\} \mbox{ is not a generalization}\big\}$. We then define $\Omega^{G_1, G_2}(A,A')$ as $(|\mathfrak{c}|+1)^{-1}$. The "+1" term is only meant to avoid division by zero.
	\end{example}
	
	A quality estimator acts as an indicator of the interest of having a couple $(A,A')$ into the generalization $\phi$ under construction. It will naturally segment the couples in $\gen(G_1,G_2)$ into subsets with different quality ($\Omega$) values, guiding our algorithm as to which couples should or should not be part of the generalization. Now, inside the main loop of Algorithm~\ref{alg:kswap}, the same estimator function can be used to guide the search for the $k$-swap - in particular the mappings $\phi_s$ and $\phi_G$ - rather than computing these by exhaustive search. 
	Algorithm~\ref{alg:selection} provides such a concrete search procedure based on $\Omega$. Given a couple of atoms $(A,A')$ and a generalization $\phi$ under construction, the algorithm searches for a suitable $\phi_s$ and $\phi_G$ that could be used as a $k$-swap to continue the construction of the generalization by Algorithm~\ref{alg:kswap}.
	
	\begin{algorithm}[hbtp]
		\caption{Selecting $\phi_s$ and $\phi_G$ for a given $(A,A')$}
		\label{alg:selection}
		\begin{algorithmic}
			\State $GS \gets \{\}$
			\State $BS \gets \{\}$
			\State $\phi_G \gets\{\}$
			\State $\phi_s \gets \phi\setminus\phi\enforce\{(A,A')\}$
			\State $S \gets \gen(G_1, G_2) \setminus\phi\enforce\{(A,A')\}$
			\While{$|\phi_G| < |\phi_s| \mbox{ and } |\phi_s|\le k$}
			\While{$|\phi_G| < |\phi_s| \mbox{ and } \neg (\compatible_{\phi\setminus\phi_s\cup\phi_G}(S) = \{\} \mbox{ and } GS = \{\})$}
			\State For all $p\in \max^W_\Omega(\compatible_{\phi\setminus\phi_s\cup\phi_G}(S)) : push(GS, (\phi_G\cup p, S\setminus\{p\}))$
			\State $(\phi_G, S)\gets pop(GS)$
			\EndWhile
			\If{$|\phi_G|<|\phi_s|$}
			\State For all $p\in \min^W_\Omega(\phi\setminus\phi_s) : enter(BS, \phi_s\cup\{p\})$
			\If{$BS\neq\{\}$}
			\State $\phi_s\gets exit(BS)$
			\State $\phi_G\gets\{\}$
			\State $S\gets \gen(G_1, G_2)\setminus(\phi\cup \{(A,A')\})$
			\Else
			\State return $\bot$
			\EndIf
			\EndIf
			\EndWhile
			\If{$|\phi_G|=|\phi_s|$}
			\State return $\phi_s, \phi_G$
			\Else
			\State return $\bot$
			\EndIf
		\end{algorithmic}
	\end{algorithm}
	
	The search process of Algorithm~\ref{alg:selection} is conceptually analogous to an A* search. The mapping $\phi_s$ is initialized with the part of $\phi$ that is incompatible with the pair of atoms $(A,A')$ we wish to enforce into the generalization. Its replacement mapping $\phi_G$ is initially empty and the algorithm subsequently searches to construct a sufficiently large $\phi_G$ (the inner while loop). During this search, $S$ represents the set of candidates, i.e. couples from $\gen(G_1,G_2)$ that are not (yet) associated to the generalization, and $\compatible_{\phi\setminus\phi_s\cup\phi_G}(S)$ represents the subset of $S$ of which each element could be added to $\phi\setminus\phi_s\cup \phi_G$ such that the result is a generalization (i.e. there is no conflict in the associated renamings). In order to explore different possibilities by backtracking, the while loop manipulates a stack $GS$ that records alternatives for $\phi_G$ with the corresponding set $S$ for further exploration. 
	
	Now, in order to steer the search process, only candidate couples having an $\Omega$-value within the best $W$ are considered for further exploration. We therefore define $\max^W_\Omega(U)$ (resp. $\min^W_\Omega(U)$) as denoting the subset of $U$ composed of those couples that have an associated $\Omega$-value among the $W$ highest (resp. lowest) qualities of elements in $U$. In this, $W$ is a parameter of the algorithm that can be used to control the degree of backtracking. If $W=\infty$ backtracking is performed over \textit{all} possible alternatives (exhaustive search), whereas when $W=1$ only the couples with the best (or worst) $\Omega$-value are considered for use. Note that even when exhaustive search is used ($W=\infty$), the algorithm considers the most promising couples (those with the highest $\Omega$-values) first. 
	
	If the search for $\phi_G$ was without a satisfying result (i.e. no $\phi_G$ is found equal in size to $\phi_s$), the algorithm continues by removing another couple from $\phi$ (thereby effectively enlarging $\phi_s$). The rationale behind this action is that there might be a couple in $\phi$ that is ``blocking'' the couples in $S$ from addition to $\phi$. In order to steer the removal of such potentially blocking couples, a couple from $\min^W_\Omega(\phi\setminus\phi_s)$ is selected, and alternatives (those having an $\Omega$-value among the $W$ worst) are recorded in a queue ($BS$). Note the use of a queue (and its associated operations \textit{enter} and \textit{exit}) as opposed to the stack $GS$.
	
	The process is repeated until either $|\phi_G|=|\phi_s|$ in what case we have found a suitable $k$-swap, or until $|\phi_s|>k$ in what case we have not, and the algorithm returns $\bot$. 
	
	\section{Prototype evaluation}\label{section-prototype}
	
	In order to experimentally evaluate both the result and performance of our approach, we have made a prototype implementation of Algorithms~\ref{alg:kswap} and~\ref{alg:selection} in Prolog\footnote{Source code is available at \url{https://github.com/Gounzy/CLPGeneralization}.}. The implementation uses the quality function $\Omega$ defined in Example~\ref{ex-omega}. Our evaluation consist in computing $k$-swap stable genera\-lizations for a considerable set of test cases (pairs of goals) that have been generated randomly according to certain criteria. In particular, we have defined 6 problem classes, the characteristics of which are represented in Table~\ref{table:classes}. 
	
	\begin{table}[hbtp]
		\caption{Classes of randomly generated anti-unification problems}
		\label{table:classes}
		\begin{tabular}{ccccccc}
			\hline
			class & Variables & Literals & Variable combinations & Literal matchings\\ 
			\hline
			1 & 5-10 & 5-15 & $\le$ 60,000 & $\le$ 40,000 \\
			2 & 6-10 & 10-15	& 60,001-360,000 & 40,001-210,000 \\
			3 & 9-10 & 15-20 & 360,001-3,600,000 & 210,001-9,000,000
			\\
			4 & 10-12 & 15-20 & 3,600,001-17,000,000 & 9,000,001-17,000,000 \\
			5 & 10-15 & 15-20 & 17,000,001-175,000,000 & 17,000,001-175,000,000 \\
			6 & 10-18 & 15-22 & 175,000,001-1,750,000,000 & 175,000,001-1,750,000,000\\
			\hline
		\end{tabular}
		\vspace{-0.5\baselineskip}
	\end{table}
	
	Table~\ref{table:classes} provides, for each problem class, a row containing the admissible (ranges of) values that were used when generating a test case $(G_1,G_2)$ belonging to that class. The columns 'Variables' and 'Literals' denote, respectively, the number of variables and literals that are allowed in the generated goals. The column 'Variable combinations' denotes the total number of mappings that must exist between the variables of $G_1$ and the variables of $G_2$. In a similar vein, the column 'Literal matchings' denotes the number of subsets of $\gen(G_1,G_2)$ (excluding those mapping a single literal more than once), as such representing an upper bound on the number of potential generalizations of $G_1$ and $G_2$. Note that these parameters (in particular the latter two) guarantee that each test case exhibits a certain complexity for the anti-unification algorithm and the parameter values of each class are chosen in such a way to have ascending complexities both with respect to the number variable combinations and literal matching possibilities that could potentially need to be explored by the algorithm. The generated literals are all atoms that are built using three test predicates $f/1, g/2$ and $h/3$. Real-life applications would typically harbor a higher number of literal symbols, but less symbols tend to increase the anti-unification complexity of the generated goals, making them more of a challenge for our algorithm. Also note that although being built on a CLP formalism, the test instances are by no means intended to depict real-life Constraint Satisfaction Problems (CSP). They rather represent batches of anti-unification instances as could arise in semantic clones detection~\cite{DBLP:conf/ppdp/MesnardPV16} where one typically needs a fast and efficient anti-unification algorithm capable of handling a multitude of goals in a reasonable time. 
	
	\begin{example}
		The following is an example of a generated test case, verifying the constraints of class 2 in Table~\ref{table:classes}. It presents 72,000 anti-unification possibilities and 181,440 possible variable combinations.
		$G_1 = \{f(A),f(C),f(F),g(C,G),g(I,E),g(I,F),h(A,A,C),h(B,F,D),h(C,A,A),h(D,E,C),\\h(F,A,C),h(F,E,H),h(G,G,B),h(G,I,I)\}\\
		G_2 = \{f(J),f(K),f(P),g(N,L),g(N,N),g(O,J),h(K,M,J),h(K,P,M)\}$
	\end{example}
	
	Table~\ref{table:results} summarizes the results of our experimental evaluation. Four incarnations of our algorithm were tested, computing $k$-swap stable generalizations for $k=0$, $k=2$, $k=4$ and $k=\infty$. Each incarnation is represented in the table by, respectively, $\Omega_0$, $\Omega_2$, $\Omega_4$ and $\Omega_\infty$. For each incarnation, we have fixed $W = 1$ in order to severely limit backtracking to alternatives having the same $\Omega$-value. While minimal backtracking is of course advantageous for the execution time, it is at the same time the most demanding setting when testing the accuracy and relevance of the $k$-swap stability concept. To compare the execution times, we have also implemented two naive brute-force algorithms, denoted in the table by $\mcg_{ER}$ and $\mcg_{EG}$, that compute an mcg either by exhaustively enumerating all possible renamings ($\mcg_{ER}$) or all possible literal matchings ($\mcg_{EG}$) and retaining the largest generalization that was thus found. 
	
	For each of the 6 problem classes, one thousand examples were generated verifying the constraints of the class. Each algorithm was executed over all 1000 examples and Table~\ref{table:results} displays their average execution time (in milliseconds). As expected, the execution time is higher for larger values of $k$, and grows with the complexity of the problems that are dealt with. However, for all classes but the simplest, the execution time of our algorithm (even in the case where $k=\infty$) stays well below the execution time of the brute-force algorithms. For the more complex problem classes, the difference amounts to several orders of magnitude and remains more than manageable (in the millisecond range), even with $k=\infty$. 
	Only for the simplest of test cases (problem class 1) our algorithm shows an overhead caused by trying out some k-swaps more than once.
	As a side note, between the two brute-force algorithms $\mcg_{ER}$ is in general the slowest because it has in general an enormous amount of variable mappings to explore, while $\mcg_{ER}$ is more often able to cut exploration paths when encountering incompatible literal matchings during its mcg construction process.

	In order to test the accuracy of our abstraction, for each example we compared the size of the computed $k$-swap stable generalization with the size of computed by the naive algorithms. For each problem class and algorithm incarnation, Table~\ref{table:results} displays the average size of the computed $k$-swap stable generalization expressed as a percentage of the size of the corresponding mcg. As can be expected, the accuracy grows for larger values of $k$ but is, on average, never below 80\% of the mcg even for the most simple and greedy incarnation of our algorithm ($\Omega_0$). Note that in the case of $\Omega_\infty$, the average accuracy is below 100\% while in theory 
	$\Omega_\infty$ should compute an mcg. This is of course due to the fact that $W=1$, meaning that not enough backtracking is performed in order to compute an mcg in all cases. These are nevertheless quite promising results.
	
	While the use of average times and accuracy might be criticized, it is noteworthy that for all problem classes and algorithms the standard deviation between the execution times was less than 20\% of the average value and less than 10\% in the case of the accuracy.
	
	\begin{table}[hbtp]
		\caption{Average execution times (in milliseconds) and average size relative to an mcg (in \%)}
		\label{table:results}
		\begin{tabular}{lrrrrrrrrrr}
			\hline
			class & $\mcg_{ER}$ & $\mcg_{EG}$ & \multicolumn{2}{c}{$\Omega_0$} & \multicolumn{2}{c}{$\Omega_2$} & \multicolumn{2}{c}{$\Omega_4$} & \multicolumn{2}{c}{$\Omega_\infty$} 
			
			\\ \hline
			
			1 & 4.66 & 1.11 & 0.48 & 97.4\% & 1.39 & 99.4\% & 2.05 & 99.9\% & 3.13 & 99.9\%
			\\
			2 & 639.15 & 154.56 & 7.76 & 84,8\% & 28.31 & 96.2\% & 56.55 & 98.3\% & 63.81 & 98.6\%
			\\
			3 & 4240 & 701.57 & 10.88 & 83,8\% & 43.55 & 95.3\% & 91.39 & 98.0\% & 104.06 & 98.2\%
			\\
			4 & 11800 & 2890 & 18.38 & 81,6\% & 71.26 & 93.9\% & 156.73 & 97.3\% & 206.22 & 97.4\%
			\\
			5 & 26150 & 7640 & 24.72 & 84.1\% & 91.07 & 94,2\% & 196.56 & 96.5\% & 249.33 & 97.5\%
			\\
			6 & 431260 & 37930 & 46.84 & 80.4\% & 127.14 & 93.4\% & 271.94 & 95.7\% & 377.2 & 96.9\%
			\\
			\hline
		\end{tabular}
		\vspace{-1\baselineskip}
	\end{table}
	In conclusion, these simple experiments show that our abstraction performs quite well: although it will in general not compute the \textit{maximal} common generalization, it will find relatively large generalizations in a tractable time (generally even impressively fast when compared to a brute-force approach), even when the overall anti-unification complexity is high. 
	
	\section{Conclusions and future work}\label{section-conclusion}
	
	In this work, we have established a theory of anti-unification (or generalization) in the context of Constraint Logic Programming. When goals are considered as sets of atoms and constraints, the problem of computing their maximal common generalization becomes an intractable problem, a result that we have formally proved. We have introduced an abstraction of the maximal common generalization, namely a $k$-swap stable generalization, that can be computed in polynomial time. We have defined a skeleton algorithm that is parametric by $k$ and that allows to steer the generalization by a heuristic function $\Omega$. We have shown our algorithm to provide promising results on a set of randomly created test cases. Its parameters should be tuned to achieve the best trade-off between output mcg size (by increasing $k$ and/or $W$) and time performance (by decreasing $k$ and/or $W$), depending on the application at hand. Future work should investigate the exact interaction between parameters $k$ and $W$: when not able to find an mcg, the responsible parameter is, in our current prototype, not clearly identified. While the heuristic function $\Omega$ we have used in our prototype implementation seems to perform quite well and results in overall large generalizations, other heuristic functions can be envisioned, possibly in function of the application at hand.

	In further work, we also aim at integrating the notions developed in this paper into a framework for clone detection or algorithmic equivalence recognition such as~\cite{DBLP:conf/ppdp/MesnardPV16} that uses CLP clauses as an intermediate program representation. Having an efficient generalization algorithm is a necessary ingredient that allows to compute the similarity between program fragments. We expect that our generalization concept and algorithm can be integrated in such a framework such that it would allow to steer the underlying transformation process. In that context, we intend to conduct a more in-depth empirical study of the two algorithms presented in Section~\ref{section-algo}. We will in particular investigate the complexity of Algorithm~\ref{alg:selection} that in practice depends on the branching factor induced by the quality estimator at hand. 
	
	Direct applications of our generalization algorithm include other transformational approaches on CLP programs, in particular those where computing generalizations is a means to obtain finiteness of the transformation, an example being partial deduction of CLP programs. 
	Our anti-unification theory is a general and domain-independent framework. As such, it can likely be incarnated and enforced by incorporating and integrating domain-specific widening operators, which is another topic for future work. Moreover, depending on the context, generalizations can be considered maximal or most-specific based on other criteria than just cardinality, a simple example being the amount of literal arguments captured in the common generalization. This is especially relevant when arities can widely vary from one literal to another, and constitutes a topic for future research on other generalization strategies.
	
	
	\section*{Acknowledgments} 
	We thank anonymous reviewers for their constructive input and remarks.
	
	\bibliographystyle{acmtrans}
	\bibliography{auclpx}

\begin{thebibliography}{}

\bibitem[\protect\citeauthoryear{Benkerimi and W.~Lloyd}{Benkerimi and
  W.~Lloyd}{1990}]{benkerimi:lloyd}
{\sc Benkerimi, K.} {\sc and} {\sc W.~Lloyd, J.} 1990.
\newblock A partial evaluation procedure for logic programs.
\newblock 343--358.

\bibitem[\protect\citeauthoryear{Bouajjani, Habermehl, Rogalewicz, and
  Vojnar}{Bouajjani et~al\mbox{.}}{2006}]{BOUAJJANI200637}
{\sc Bouajjani, A.}, {\sc Habermehl, P.}, {\sc Rogalewicz, A.}, {\sc and} {\sc
  Vojnar, T.} 2006.
\newblock Abstract regular tree model checking.
\newblock {\em Electronic Notes in Theoretical Computer Science\/}~{\em
  149,\/}~1, 37 -- 48.
\newblock Proceedings of the 7th International Workshop on Verification of
  Infinite-State Systems (INFINITY 2005).

\bibitem[\protect\citeauthoryear{Burghardt}{Burghardt}{2014}]{DBLP:journals/corr/Burghardt14b}
{\sc Burghardt, J.} 2014.
\newblock E-generalization using grammars.
\newblock {\em CoRR\/}~{\em abs/1403.8118}.

\bibitem[\protect\citeauthoryear{{De Schreye}, Glück, Jørgensen, Leuschel,
  Martens, and Sørensen}{{De Schreye} et~al\mbox{.}}{1999}]{DESCHREYE1999231}
{\sc {De Schreye}, D.}, {\sc Glück, R.}, {\sc Jørgensen, J.}, {\sc Leuschel,
  M.}, {\sc Martens, B.}, {\sc and} {\sc Sørensen, M.~H.} 1999.
\newblock Conjunctive partial deduction: foundations, control, algorithms, and
  experiments.
\newblock {\em The Journal of Logic Programming\/}~{\em 41,\/}~2, 231 -- 277.

\bibitem[\protect\citeauthoryear{Fioravanti, Pettorossi, Proietti, and
  Senni}{Fioravanti et~al\mbox{.}}{2013}]{DBLP:journals/fuin/FioravantiPPS13}
{\sc Fioravanti, F.}, {\sc Pettorossi, A.}, {\sc Proietti, M.}, {\sc and} {\sc
  Senni, V.} 2013.
\newblock Controlling polyvariance for specialization-based verification.
\newblock {\em Fundam. Inform.\/}~{\em 124,\/}~4, 483--502.

\bibitem[\protect\citeauthoryear{Gallagher}{Gallagher}{1993}]{Gallagher:1993:TSL:154630.154640}
{\sc Gallagher, J.~P.} 1993.
\newblock Tutorial on specialisation of logic programs.
\newblock In {\em Proceedings of the 1993 ACM SIGPLAN Symposium on Partial
  Evaluation and Semantics-based Program Manipulation}. PEPM '93. ACM, New
  York, NY, USA, 88--98.

\bibitem[\protect\citeauthoryear{Gange, Navas, Schachte, S{\o}ndergaard, and
  Stuckey}{Gange et~al\mbox{.}}{2015}]{DBLP:journals/tplp/GangeNSSS15}
{\sc Gange, G.}, {\sc Navas, J.~A.}, {\sc Schachte, P.}, {\sc S{\o}ndergaard,
  H.}, {\sc and} {\sc Stuckey, P.~J.} 2015.
\newblock Horn clauses as an intermediate representation for program analysis
  and transformation.
\newblock {\em {TPLP}\/}~{\em 15,\/}~4-5, 526--542.

\bibitem[\protect\citeauthoryear{Guti{\'e}rrez-Naranjo, Alonso-Jim{\'e}nez, and
  Borrego-D{\'i}az}{Guti{\'e}rrez-Naranjo et~al\mbox{.}}{2003}]{subsumption}
{\sc Guti{\'e}rrez-Naranjo, M.~A.}, {\sc Alonso-Jim{\'e}nez, J.~A.}, {\sc and}
  {\sc Borrego-D{\'i}az, J.} 2003.
\newblock Generalizing programs via subsumption.
\newblock In {\em Computer Aided Systems Theory - EUROCAST 2003},
  {R.~Moreno-D{\'i}az} {and} {F.~Pichler}, Eds. Springer Berlin Heidelberg,
  Berlin, Heidelberg, 115--126.

\bibitem[\protect\citeauthoryear{Jaffar, Maher, Marriott, and Stuckey}{Jaffar
  et~al\mbox{.}}{1998}]{clp-semantics}
{\sc Jaffar, J.}, {\sc Maher, M.}, {\sc Marriott, K.}, {\sc and} {\sc Stuckey,
  P.} 1998.
\newblock The semantics of constraint logic programs.
\newblock {\em The Journal of Logic Programming\/}~{\em 37,\/}~1, 1 -- 46.

\bibitem[\protect\citeauthoryear{Jaffar and Maher}{Jaffar and
  Maher}{1994}]{clpsurvey}
{\sc Jaffar, J.} {\sc and} {\sc Maher, M.~J.} 1994.
\newblock Constraint logic programming: a survey.
\newblock {\em The Journal of Logic Programming\/}~{\em 19-20}, 503 -- 581.
\newblock Special Issue: Ten Years of Logic Programming.

\bibitem[\protect\citeauthoryear{Khardon and Arias}{Khardon and
  Arias}{2006}]{subsumption-lattice}
{\sc Khardon, R.} {\sc and} {\sc Arias, M.} 2006.
\newblock The subsumption lattice and query learning.
\newblock {\em Journal of Computer and System Sciences\/}~{\em 72,\/}~1, 72 --
  94.

\bibitem[\protect\citeauthoryear{Leuschel, Martens, and Schreye}{Leuschel
  et~al\mbox{.}}{1998}]{DBLP:journals/toplas/LeuschelMS98}
{\sc Leuschel, M.}, {\sc Martens, B.}, {\sc and} {\sc Schreye, D.~D.} 1998.
\newblock Controlling generalization amd polyvariance in partial deduction of
  normal logic programs.
\newblock {\em {ACM} Trans. Program. Lang. Syst.\/}~{\em 20,\/}~1, 208--258.

\bibitem[\protect\citeauthoryear{Mesnard, Payet, and Vanhoof}{Mesnard
  et~al\mbox{.}}{2016}]{DBLP:conf/ppdp/MesnardPV16}
{\sc Mesnard, F.}, {\sc Payet, {\'{E}}.}, {\sc and} {\sc Vanhoof, W.} 2016.
\newblock Towards a framework for algorithm recognition in binary code.
\newblock In {\em Proceedings of the 18th International Symposium on Principles
  and Practice of Declarative Programming, Edinburgh, United Kingdom, September
  5-7, 2016}, {J.~Cheney} {and} {G.~Vidal}, Eds. {ACM}, 202--213.

\bibitem[\protect\citeauthoryear{Pettorossi and Proietti}{Pettorossi and
  Proietti}{1998}]{DBLP:journals/csur/PettorossiP98}
{\sc Pettorossi, A.} {\sc and} {\sc Proietti, M.} 1998.
\newblock Program specialization via algorithmic unfold/fold transformations.
\newblock {\em {ACM} Comput. Surv.\/}~{\em 30,\/}~3es, 6.

\bibitem[\protect\citeauthoryear{Plotkin}{Plotkin}{1970}]{plotkin}
{\sc Plotkin, G.~D.} 1970.
\newblock A note on inductive generalization.
\newblock {\em Machine Intelligence\/}~{\em 5}, 153--163.

\bibitem[\protect\citeauthoryear{S{\o}rensen and Gl\"{u}ck}{S{\o}rensen and
  Gl\"{u}ck}{1999}]{Sorensen:1998:IS:645795.665914}
{\sc S{\o}rensen, M.~H.} {\sc and} {\sc Gl\"{u}ck, R.} 1999.
\newblock Introduction to supercompilation.
\newblock In {\em Partial Evaluation - Practice and Theory, DIKU 1998
  International Summer School}. Springer-Verlag, Berlin, Heidelberg, 246--270.

\bibitem[\protect\citeauthoryear{Sørensen and Glück}{Sørensen and
  Glück}{1995}]{Sorensen95analgorithm}
{\sc Sørensen, M.~H.} {\sc and} {\sc Glück, R.} 1995.
\newblock An algorithm of generalization in positive supercompilation.
\newblock In {\em Proceedings of ILPS'95, the International Logic Programming
  Symposium}. MIT Press, 465--479.

\bibitem[\protect\citeauthoryear{Sys{\l}o}{Sys{\l}o}{1982}]{SYSLO198291}
{\sc Sys{\l}o, M.~M.} 1982.
\newblock The subgraph isomorphism problem for outerplanar graphs.
\newblock {\em Theoretical Computer Science\/}~{\em 17,\/}~1, 91 -- 97.

\end{thebibliography}
	
\end{document}